\documentclass[aps,prl,twocolumn,superscriptaddress,floatfix,nofootinbib,showpacs,longbibliography,groupedaddress]{revtex4-1}

\usepackage[utf8]{inputenc}  
\usepackage[T1]{fontenc}     
\usepackage[table]{xcolor}
\usepackage[british]{babel}  
\usepackage[sc,osf]{mathpazo}
\usepackage[scaled=0.86]{berasans}  
\usepackage[colorlinks=true, citecolor=blue, urlcolor=blue]{hyperref}  
\usepackage{graphicx} 
\usepackage[babel]{microtype}  
\usepackage{amsmath,amssymb,amsthm,bm,amsfonts,mathrsfs,bbm} 

\usepackage{xspace}  
\usepackage{pgf,tikz}
\usepackage{xcolor}
\usepackage{multirow}
\usepackage{array}
\usepackage{bigstrut}
\usepackage{braket}
\usepackage{color}
\usepackage{natbib}
\usepackage{multirow}
\usepackage{mathtools}
\usepackage{float}
\usepackage[caption = false]{subfig}
\usepackage{xcolor,colortbl}
\usepackage{color}
\usepackage{rotating}

\newcommand{\be}{\begin{equation}}
\newcommand{\ee}{\end{equation}}
\newcommand{\ba}{\begin{eqnarray}}
\newcommand{\ea}{\end{eqnarray}}

\newtheorem{theorem}{Theorem}

\def\>{\rangle}
\def\<{\langle}

\usepackage{centernot}
\usepackage{subfig}

\begin{document}

\title{From no causal loop to absoluteness of cause: discarding the quantum NOT logic}

\author{Anandamay Das Bhowmik}
\affiliation{Physics and Applied Mathematics Unit, Indian Statistical Institute, 203 BT Road, Kolkata, India.}

\author{Preeti Parashar}
\affiliation{Physics and Applied Mathematics Unit, Indian Statistical Institute, 203 BT Road, Kolkata, India.}

\author{Guruprasad Kar}
\affiliation{Physics and Applied Mathematics Unit, Indian Statistical Institute, 203 BT Road, Kolkata, India.}

\author{Manik Banik}
\affiliation{School of Physics, IISER Thiruvananthapuram, Vithura, Kerala 695551, India.}

\begin{abstract}
The principle of `absoluteness of cause' (AC) assumes the cause-effect relation to be observer independent and is a distinct assertion than prohibiting occurrence of any causal loop. Here, we study implication of this novel principle to derive a fundamental no-go result in quantum world. AC principle restrains the `time order' of two spacelike separated events/processes to be a potential cause of another event in their common future, and in turn negates existence of a quantum device that transforms an arbitrary pure state to its orthogonal one. The present {\it no-go} result is quite general as its domain of applicability stretches out from the standard linear quantum theory to any of its generalizations allowing deterministic or stochastic nonlinear evolution. We also analyze different possibilities of violating the AC principle in generalized probability theory framework. A strong form of violation enables instantaneous signaling, whereas a weak form of violation forbids the theory to be locally tomographic. On the other hand, impossibility of an intermediate violation suffices to discard the universal quantum NOT logic.    
\end{abstract}



\maketitle
{\it Introduction.--}  No-go principles hypothesized from empirical or/and metaphysical stands play crucial roles in constructing theories of the physical world. For instance, the second law of thermodynamics is a no-go that isolates unfavored directions of action among the physical processes \cite{Lieb99}, whereas the principle of relativistic causality places an embargo on the occurrence of any causal loop in the spacetime structure so that no grandfather kind of paradox arises \cite{Hawking92,Horodecki19}. While these principles have general appeal and are considered to be true across different possible theories, more specific no-go results can be derived while considering a particular theory. Bell and Kochen–Specker no-go theorems are seminal such examples in quantum foundations that exclude certain classical like interpretations of quantum predictions \cite{Bell66,Kochen67,Mermin93}. On the other hand, advent of quantum information theory identifies a second type of no-go results, such as no-cloning theorem \cite{Wootters82}, no-flipping theorem \cite{Massar95,Gisin99,Enk05}, no-broadcasting theorem \cite{Barnum96}, no-programming theorem \cite{Nielsen97}, no-deleting theorem \cite{Pati00}, no-hiding theorem \cite{Braunstein07}, that follow algebraically from the Hilbert space formulation of quantum theory. Naturally the question arises how robust these second type of no-go theorems are, {\it i.e.} if quantum theory needs to be modified by some more general theory what will be the status  of these no-go theorems in the modified theory. This question motivates re-derivation of these second kind of no-go's from more general theory independent physical principles. Interestingly, it has been shown that the no-cloning and no-deleting theorems can be derived from the assumption of no instantaneous signalling conditions \cite{Gisin98,Hardy99,Pati03}. With a similar aspiration the authors in \cite{Chattopadhyaya06} have tried to obtain the no-flipping theorem starting from some general principle. However, their proof presupposes further structure on quantum state evolution which is an additional demand to a general flipping device. The proof in \cite{Chattopadhyaya06}, therefore, lacks some generality and motivates further look towards principle based derivation of no-flipping theorem. 

In this work we first analyze a causal principle, namely the principle of absoluteness of cause (AC) that assumes the cause-effect relation to be observer independent. Interestingly this principle has operational implications in generalized physical theories. While considering more than two events/processes, the AC principle prohibits the `time order' of two spacelike separated events/processes to be a potential cause of another event in their common causal future. Quite fascinatingly this implication of AC helps us to establish the quantum no-flipping theorem without assuming any further structure on quantum state evolution. The present no-go theorem, therefore, not only applies to  standard textbook quantum theory that assumes linear Schr\"{o}dinger dynamics, rather it holds true in any generalized version of quantum theory as proposed by Weinberg, Pearle, Gisin, Ghirardi, Diosi, Polchinski, and others that allow nonlinear evolution rule either deterministic or stochastic \cite{Pearle76,Gisin81,Gisin84,Ghirardi86,Diosi88(1),Diosi88(2),Weinberg89(1),Weinberg89(2),Gisin89,Polchinski91,Gisin95,Weinberg12}. We further analyze different possibilities of violating the AC principle in generalized probabilistic theories \cite{Barrett07,Chiribella10,Masanes11,Hardy11,Hardy12}. It turns out that, a `strong' form of violation makes instantaneous signaling possible. On the other hand, a `weak' violation of AC principle, in an operational theory, makes the theory incompatible with {\it tomographic locality} condition, which is postulated to hold true in different approaches of axiomatic reconstruction of quantum theory  \cite{Barrett07,Chiribella10,Masanes11,Hardy11,Hardy12}. In quantum theory, allowing linear as well as deterministic/ stochastic nonlinear evolution, impossibility of an intermediate violation, lying in between the weak and strong ones, suffices to discard the universal NOT operation.

{\it The principle of absoluteness of cause.--} 
Causality deals with relationship between causes and effects and plays crucial role in development of the theory of knowledge and shapes our worldview \cite{Pearl09}. In generic sense, a cause can be an event, a process, space-time relations among several events or processes (for instance, their 'time order'), or physical interactions between systems, or any other variable that is responsible (either partially or fully) for the occurrence of another set of events, processes or space-time relations thereof or correlations, generally called an effect. The principle of no causal loop (NCL)  prohibits some disagreements among different observers regarding cause-effect relation. The statement of NCL can be formulated as the following,
\begin{itemize}
\item[ NCL:] If $A$ is the cause of $B$ for some observer, then for any other observer, $B$ cannot be a cause of $A$.
\end{itemize} 
The framework of special theory of relativity, where events are described by spacetime coordinates associated with different inertial observers is compatible with NCL since light-cone ordering remains invariant under Lorentz boost, which in tern, suffices to prevent causal loops and thereby rules out the possibility of grandfather kind of paradoxes. \cite{Self1,Godel49}. Notably, in bipartite Bell scenario the aforesaid statement has shown to be equivalent to the no-signalling conditions, whereas it is a strictly weaker demand than no-signalling conditions in multipartite scenario \cite{Horodecki19}. 
However, the assertion of the principle of absoluteness of cause (AC) is different from NCL as it prohibits a different kind of disagreement among different observers.
\begin{itemize}
\item[AC:]  If $A$ is the cause of $B$ for some observer, then $A$ will remain as the cause of $B$ for all observers.
\end{itemize}
AC principle demands the cause to be observer independent and hence absolute. It forbids a different kind of paradox, where, a sculptor creates a sculpture in some reference frame, whereas to an observer in some other frame the sculpture is created without being caused by its creator. Clearly, no causal loop is formed in this paradox. In this work our aim is to study operational implications of this causal principle. 
\begin{figure}[t!]
\centering
\includegraphics[width=0.45\textwidth]{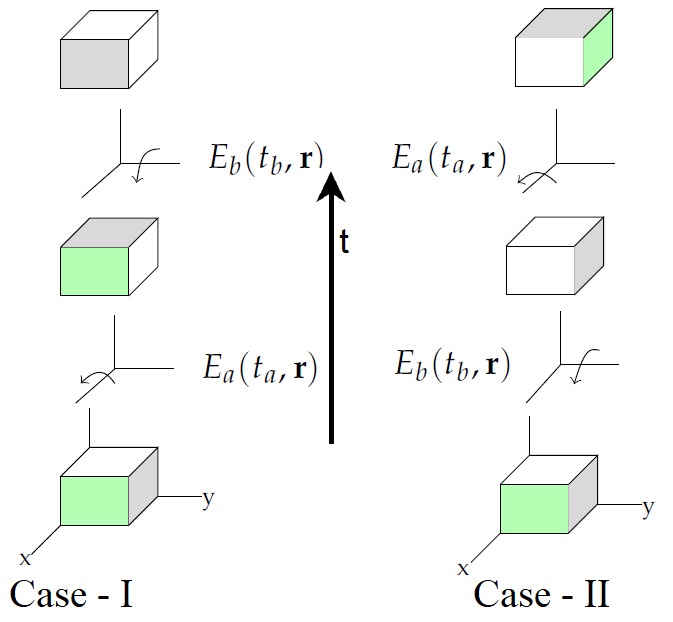}
\caption{(Color online) Temporal orders of timelike separated processes can lead to distinct futures. Let, $\mathrm{E}_a$ and $\mathrm{E}_b$ denote $90^\circ$ rotation along $x$-axis and $y$-axis, respectively. Consider a cube with different colored surfaces. Starting from the same initial configuration, its final configurations differ depending on whether $\mathrm{E}_a$ precedes (succeeds) $\mathrm{E}_b$. Case-I and case-II depict two different final configurations.}\label{fig1}
\end{figure}

The principle of AC puts nontrivial restriction on potential causes. Usually cause is assumed to be associated with events in spacetime or processes carried out in sufficiently small region of spacetime.  However, in the scenarios involving multiple processes or/and events the spacetime relation between those events or processes (for instance, their `time order') can also be the cause. In such a scenario AC principle brings nontrivial implications. Consider two timelike separated processes $\mathrm{E}_a(t_a,{\bf r}_a)$ and $\mathrm{E}_b(t_b,{\bf r}_b)$, in some reference frame $\mathcal{I}\equiv(t,\mathbf{r})$. Following two temporal orders are possible: (i) $t_a<t_b$, {\it i.e.} $\mathrm{E}_a(t_a,{\bf r}_a)$ lies in the past light cone of $\mathrm{E}_b(t_b,{\bf r}_b)$, (ii) $t_b<t_a$, {\it i.e.} $\mathrm{E}_a(t_a,{\bf r}_a)$ lies in the future light cone of $\mathrm{E}_b(t_b,{\bf r}_b)$. These two different temporal orders can lead to distinct futures, {\it i.e.} depending on which of these two possible orders is realized, distinct observations are possible in their common future. One such example is discussed in Figure \ref{fig1}. Being timelike separated the temporal order of these two processes are observer independent and hence AC principle does not prohibit their time-order to be a cause of the events in their common future. Consider now $\mathrm{E}_a(t_a,{\bf r}_a)$ and $\mathrm{E}_b(t_b,{\bf r}_b)$ to be spacelike separated. The temporal order between them becomes observer dependent - there exist different inertial frames where all the different possibilities can arise (see Figure \ref{fig2}). The principle of AC prohibits the possible time orders of such spacelike separated processes/events to be a potential cause of another event at their common future. Time order being observer dependent, in this case, lacks absoluteness and hence cannot be a potential cause. In the following we will employ this to exclude the possibility of a universal quantum NOT operation, {\it i.e.} a universal flipping device.
\begin{figure}[t!]
\centering
\includegraphics[width=0.48\textwidth]{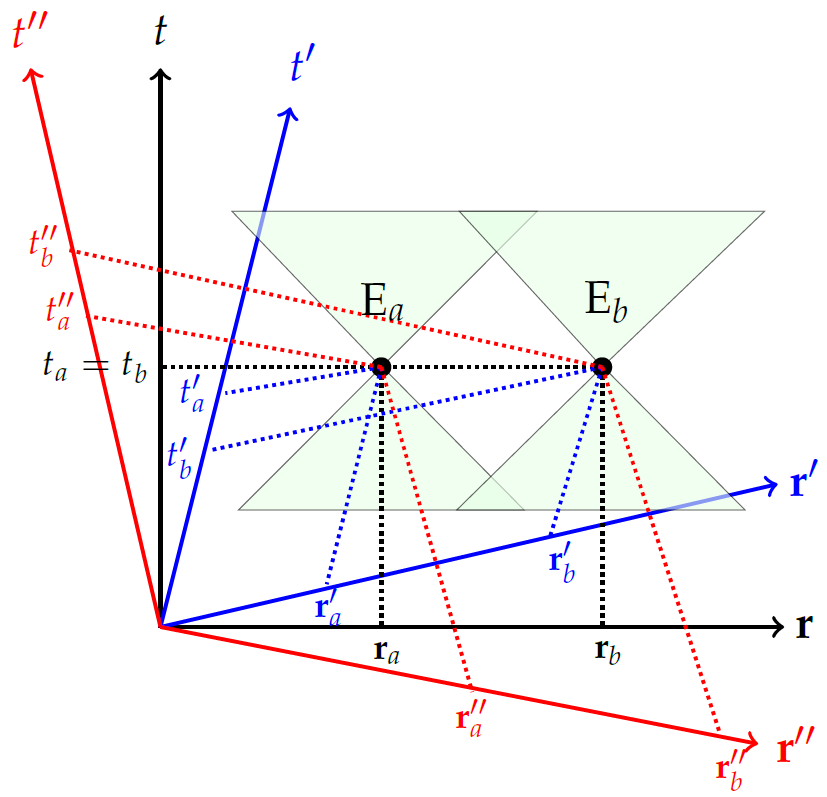}
\caption{(Color online) Different possible temporal orders between two spacelike separated events or processes. $\mathrm{E}_a$ and $\mathrm{E}_b$ are simultaneous in reference frame $\mathcal{I}\equiv(t,{\bf r})$; to an observer in reference frame $\mathcal{I}^\prime$, $\mathrm{E}_b$ occurs earlier than $\mathrm{E}_a$; on the other hand, in reference frame $\mathcal{I}^{\prime\prime}$ , $\mathrm{E}_a$ precedes $\mathrm{E}_b$. The principle of AC prohibits these different possible time orderings of the $\mathrm{E}_a$ and $\mathrm{E}_b$ to be the potential cause of some event at their common future.}\label{fig2}
\end{figure}

{\it No universal quantum flipper.--} The demand from a quantum flipping device is to get an orthogonal state $\ket{\psi^\perp}$ as output for every input state $\ket{\psi}$. If we consider linear evolution as described by Schr\"{o}dinger equation, impossibility of such a device is immediate \cite{Massar95,Gisin99,Enk05}. Some fundamental considerations, however, suggest nonlinear generalization of the Schr\"{o}dinger evolution \cite{Mielnik74,Birula76,Kibble78,Weinberg89(1),Weinberg89(2),Doebner92}. In particular, unification of the Schr\"{o}dinger dynamics and the nonlinear  projection dynamics demands modification of the deterministic linear Schr\"{o}dinger evolution rule. While the deterministic nonlinear generalizations as proposed by Weinberg \cite{Weinberg89(1),Weinberg89(2)} have shown to be inconsistent with the second law of thermodynamics \cite{Peres89} and relativistic causality principle \cite{Gisin90,Polchinski91,Czachor91}, invoking appropriate stochasticity such models can be made consistent    
\cite{Pearle76,Gisin81,Gisin84,Ghirardi86,Diosi88(1),Diosi88(2),Gisin89,Gisin95,Weinberg12} (see also \cite{Simon01}). All these models motivate a re-look to the possibility of a universal flipping device there. Without considering any particular evolution rule, in the following we, however, prove a generic no-go result.
\begin{theorem}\label{theo1}
Any formulation of quantum theory satisfying the principle of AC will not allow a universal flipping device. 	
\end{theorem}
\begin{proof}
Let Alice and Bob are two spacelike separated parties sharing a bipartite two-level quantum system prepared in the composite state $\rho_{ab} = \ket{\psi^-}_{ab}\bra{\psi^-}$, where $\ket{\psi^-}_{ab}:= (\ket{0}_a\ket{1}_b-\ket{1}_a\ket{0}_b)/\sqrt{2}$; $\{\ket{0},\ket{1}\}$ be the two normalized eigenkets of Pauli $\sigma_z$ operator. The state can be prepared at their common past and then Alice and Bob become far separated with their respective subsystems, as considered in seminal Einstein-Rosen-Podolsky (EPR) thought experiment \cite{EPR35}. In contrary to the statement of the Theorem, assume there exists a universal flipping device $\mathrm{F}$ and Bob possesses this machine, with which we will analyze the following cases. 

{\it Case-I:} Suppose in some inertial reference frame Alice first performs $\sigma_z$ measurement on her part of the composite system and then Bob applies the flipping device on his part of the system. After Alice's measurement, due to the wavefunction collapse, the composite system evolves to 
\footnotesize
$$\sigma_{ab}=\frac{1}{2}(\ket{0}_a\bra{0}\otimes\ket{1}_b\bra{1}+\ket{1}_a\bra{1}\otimes\ket{0}_b\bra{0}),$$
\normalsize
and after the action of Bob's flipping device the state becomes
\footnotesize
$$\eta_{ab}=\frac{1}{2}(\ket{0}_a\bra{0}\otimes\ket{0}_b\bra{0}+\ket{1}_a\bra{1}\otimes\ket{1}_b\bra{1}).$$
\normalsize
Note that the state update after the action of $\mathrm{F}$ simply follows from the definition of the device, and importantly, unlike other proofs of no-flipping theorem, linearity is not assumed at this step. Now in the same inertial frame consider the opposite time ordering of the concerned operations, {\it i.e.} Bob first applies $\mathrm{F}$ on his part and then Alice performs $\sigma_z$ measurement on her part. Let after the action of Bob the updated state is $\chi_{ab}$. Since we have not put any constraint on the action of $\mathrm{F}$, other than its definition, the form of the state $\chi_{ab}$ allows a vast possibilities. At this point AC principle plays crucial role. According to this principle after Alice's measurement on her part of $\chi_{ab}$ the composite system should be same as $\eta_{ab}$. Otherwise, time ordering of the spacelike separated events (Alice's $\sigma_z$ measurement and Bob's application of the flipping device $\mathrm{F}$) will lead to an observable distinction, which is prohibited by AC. The requirement of $AC$ principle thus constrains the solution for $\chi_{ab}$ to lie in
\footnotesize
\begin{align}\label{eq1}
\left\{\!\begin{aligned}
\mbox{Ch}\left\{\ket{\epsilon_k}\bra{\epsilon_k}\right\},~~~~~~~~~~~~~~~~~~~~~~~~~~~~\\
\ket{\epsilon_k}:=\frac{1}{\sqrt{2}}\left(\ket{0}_a\ket{0}_b+e^{i\alpha_k}\ket{1}_a\ket{1}_b\right)~\&~\alpha_k\in[0,2\pi)
\end{aligned}\right\},	
\end{align}
\normalsize
where $\mbox{Ch}$ denotes the convex hull of a set.

{\it Case-II:} This case is similar to {\it Case-I}, but Alice's measurement $\sigma_z$ replaced with the measurement $\sigma_x$. While Bob first applies the flipping device on his part of the composite state $\ket{\psi^-}_{ab}$ and then Alice performs $\sigma_x$ measurement on her part then, in accordance with the AC principle, the solution for $\chi_{ab}$ lie within
\footnotesize
\begin{align}\label{eq2}
\left\{\!\begin{aligned}
\mbox{Ch}\left\{\ket{\xi_k}\bra{\xi_k}\right\},~~~~~~~~~~~~~~~~~~~~~~~~~~~~\\
\ket{\xi_k}:=\frac{1}{\sqrt{2}}\left(\ket{x}_a\ket{x}_b+e^{i\beta_k}\ket{\bar{x}}_a\ket{\bar{x}}_b\right)~\&~\beta_k\in[0,2\pi)
\end{aligned}\right\}.	
\end{align}
\normalsize   
Here, $\ket{x}:=(\ket{0}+\ket{1})/\sqrt{2}~\&~\ket{\bar{x}}:=(\ket{0}-\ket{1})/\sqrt{2}$ be the eigenkets of Pauli $\sigma_x$ operator. Importantly, the set (\ref{eq1}) and set (\ref{eq2}) have only one common solution which is the state $\ket{\phi^+}_{ab}:=(\ket{0}_a\ket{0}_b+\ket{1}_a\ket{1}_b)/\sqrt{2}$, corresponding to $\alpha_k=0$ in (\ref{eq1}) and $\beta_k=0$ in (\ref{eq2}) (see Figure \ref{fig3}).
\begin{figure}[t!]
\centering
\includegraphics[width=0.4\textwidth]{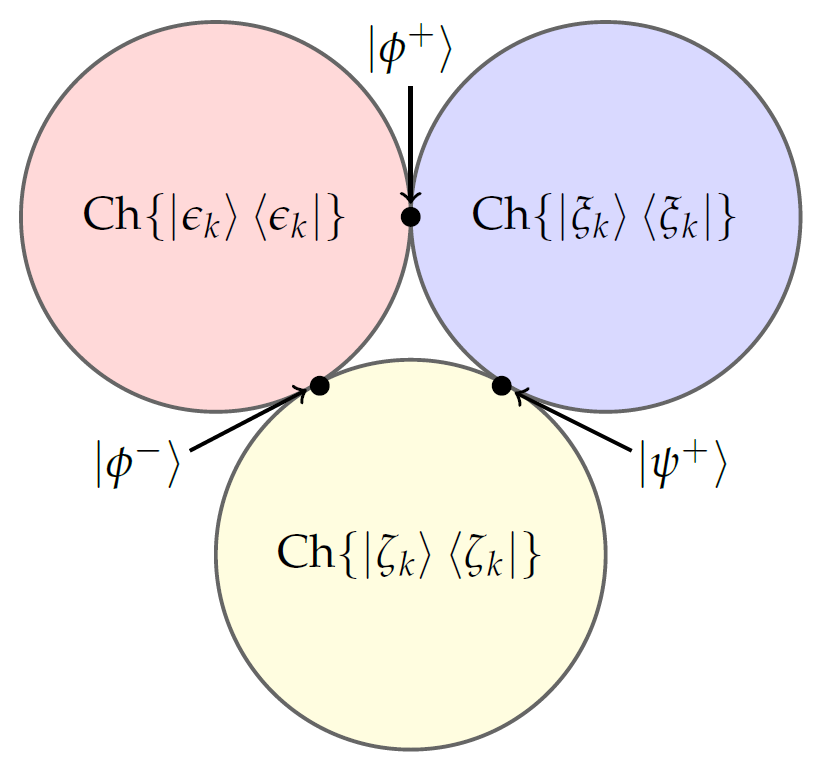}
\caption{(Color online) The circular regions denote the set of solutions for $\chi_{AB}$ obtained as a consequence of AC in the Cases I, II and III. Note that,  $\mbox{Ch}\{\ket{\epsilon_k}\bra{\epsilon_k}\}\cap\mbox{Ch}\{\ket{\xi_k}\bra{\xi_k}\}=\ket{\phi^+}\bra{\phi^+}$, $\mbox{Ch}\{\ket{\epsilon_k}\bra{\epsilon_k}\}\cap\mbox{Ch}\{\ket{\zeta_k}\bra{\zeta_k}\}=\ket{\phi^-}\bra{\phi^-}$, and $\mbox{Ch}\{\ket{\zeta_k}\bra{\zeta_k}\}\cap\mbox{Ch}\{\ket{\xi_k}\bra{\xi_k}\}=\ket{\psi^+}\bra{\psi^+}$; where $\ket{\psi^\pm}:=\frac{1}{\sqrt{2}}(\ket{0}\ket{1}\pm\ket{1}\ket{0})~\&~\ket{\phi^\pm}:=\frac{1}{\sqrt{2}}(\ket{0}\ket{0}\pm\ket{1}\ket{1})$. Furthermore, $\mbox{Ch}\{\ket{\epsilon_k}\bra{\epsilon_k}\}\cap\mbox{Ch}\{\ket{\xi_k}\bra{\xi_k}\}\cap\mbox{Ch}\{\ket{\zeta_k}\bra{\zeta_k}\}=\emptyset$, which discards the universal NOT in standard as well as in any general formulation of quantum theory.}\label{fig3}
\end{figure}

{\it Case-III:} Consider now a situation where Bob possesses the device $\mathrm{F}$ and Alice has the measuring device that performs Pauli $\sigma_y$ measurement. Similar arguments, likewise the earlier two cases, allow following possibilities for the updated state $\chi_{ab}$ after the action of Bob's flipping device,
\footnotesize
\begin{align}\label{eq3}
\left\{\!\begin{aligned}
\mbox{Ch}\left\{\ket{\zeta_k}\bra{\zeta_k}\right\},~~~~~~~~~~~~~~~~~~~~~~~~~~~~\\
\ket{\zeta_k}:=\frac{1}{\sqrt{2}}\left(\ket{y}_a\ket{y}_b+e^{i\delta_k}\ket{\bar{y}}_a\ket{\bar{y}}_b\right)~\&~\delta_k\in[0,2\pi)
\end{aligned}\right\},	
\end{align}
\normalsize  
where $\ket{y}:=(\ket{0}+i\ket{1})/\sqrt{2}~\&~\ket{\bar{y}}:=(\ket{0}-i\ket{1})/\sqrt{2}$ be the eigenkets of Pauli $\sigma_y$ operator. The set in (\ref{eq3}) does not contain the state $\ket{\phi^+}$ which happens to be the common solution for the earlier two cases. In other words, these three different considerations do not allow any common solution and hence discard the possibility of the universal flipping device $\mathrm{F}$. This completes the proof.     
\end{proof}
Few comments are in order. Firstly, in our proof we have assumed existence of the state $\ket{\psi^-}$, which also appears in many other works, such as Bohm version of the seminal EPR paradox \cite{Bohm57}, experimental establishment of the fact that quantum theory is not compatible with {\it local-realist} description \cite{Aspect81} {\it etc}. Recent developments in quantum information theory, however, allow self-testing of this state prepared from some un-characterized device \cite{Mayers04,McKague12}. 
Our argument furthermore invokes the wavefunction update due to the projective measurement. At this point we note that the wavefunction update has also been assumed by Gisin while deriving no-cloning theorem from no-signalling condition \cite{Gisin98} as well as while establishing superluminal communications in Weinberg's non-linear quantum mechanics \cite{Gisin90}. In fact, Polchinski had pointed out the crucial role of projection postulate in the second case \cite{Polchinski91}.  

{\it Implications of AC in GPT.--} The framework of generalized probability theory (GPT) confines a broad class of theories that use the notion of states to yield the outcome probabilities of measurements \cite{Barrett07,Chiribella10,Masanes11,Hardy11,Hardy12}. While an elementary system is specified by the tuple $(\Omega, \mathcal{E})$ of state and effect spaces respectively, a bipartite system $(\Omega_{ab}, \mathcal{E}_{ab})$ composed of elementary systems $(\Omega_a, \mathcal{E}_a)$ and $(\Omega_b, \mathcal{E}_b)$ is assumed to satisfy the no-signaling and local tomography condition's \cite{Self2}. A state $\omega_{ab}\in\Omega_{ab}$ is called entangled if it is not a convex mixture of product states, and such a state shared between two spatially separated parties (say Alice and Bob) can only be created in their common past. In this sense, entangled states possess `holistic' feature. Tomographic locality imposes `limited holism' \cite{Hardy12} as such states can be completely identified by comparing the local statistics of Alice and Bob at some point in their common causal future. 

The principle of AC, when applies to a GPT, prohibits the different time orders of spacelike separated local operations on a bipartite state $\omega_{ab}$ to evolve it into different states. In a GPT, this particular principle can be violated in different ways. In the strongest form of violation, different time orders of spacelike events/processes on composite state result in different final states that Alice or/and Bob can distinguish from her or/and his local statistics. However, such a violation enables instantaneous communication from Bob to Alice (or/and Alice to Bob) and thus contradicts relativistic causality principle. In the weakest form of violation, the resulting states cannot be distinguished neither by Alice nor by Bob locally. Even-more, the states cannot be distinguished by comparing the local statistics at some common future of Alice and Bob. However, it might happen that some joint measurement (eg. Bell basis measurement in quantum theory or its generalization considered in GPT \cite{Czekaj18}) on the bipartite state leads to a distinction between two possible states and hence violates the AC principle. If such a violation happens, the theory will not be tomographic local anymore, and hence will allow ample amount of holism as opposed to `limited holism'. In other words, the AC principle appends adequate justification to the tomographic locality condition from a causal perspective.  

In between weak and strong violation, there is a scope of intermediate violation which contradicts neither relativistic causality nor tomographic locality. Depending on the time order of Alice and Bob's spacelike separated operations on a bipartite state $\omega_{ab}$, the future joint states can be different in a way that yield identical marginal statistics but correlations between these local measurements are different. Therefore, unlike the weak violation, the intermediate one does not require any joint measurement to be performed on the bipartite system to detect the time-order of Alice and Bob's spacelike separated operations. For such a violation, a given joint state $\omega_{ab}$ will result into two different joint input-output probability distributions for a given pair of local measurements. Now, for a theory allowing signalling correlation, freedom of choice of the observers enforces the state to carry `excess baggage' of probability distributions -- there will be two different joint probability distributions for a given set of local measurements performed on the composite state \cite{Bub14}. Such a requirement, however, is not there in no-signalling theories. But, a no-signalling theory allowing intermediate violation of AC, has to bear this `excess baggage' of probability distributions. Note that, in Theorem \ref{theo1}, impossibility of this intermediate violation is sufficient to deduce the desired conclusion.

{\it Discussions.--} We have analyzed a novel primitive in causal framework, namely, the principle of absoluteness of cause, and discuss implications of its violation in generalized operational theory. Interestingly, this primitive yields some causal justification to the postulate of tomographic locality which is considered as an axiom in the physical reconstruction of Hilbert space quantum theory. Using this causal primitive we also derive a basic no-go result which holds true in any generalized formulation of quantum theory, possibly allowing stochastic nonlinear evolution. Apart from the foundational interests, researchers have also explored computational power of such generalized models \cite{Abrams98,Aaronson05,Bennett09,Childs16}. Our no-go theorem establishes that any such generalized model will not approve perfect universal NOT logic if it is meant to be compatible with the basic primitive of absoluteness of cause. Few questions are in order for further research. First of all, it might be interesting to obtain some bound on the degree of imperfection on NOT operation starting from the AC principle \cite{Martini02,Gisin02,Werner99}. It will be worth exploring to see what kind of general restriction(s) the AC principle imposes on the dynamics in any generic formulation of quantum theory. It might also be interesting to explore connection between this new causal primitive with the second law of thermodynamics and study its implication in spacetime structures that are more exotic than allowed in special theory of relativity. 

{\bf Acknowledgement:} It is our pleasure to thank Prof. Nicolas Gisin and Prof. Michael J. W. Hall for their insightful comments on the earlier version of the manuscript. MB acknowledges support through the research grant of INSPIRE Faculty fellowship from the Department of Science and Technology, Government of India.


\begin{thebibliography}{99}
\bibitem{Self1} While NCL is respected in a special relativistic framework, causal loops are possible in general relativity e.g., Gödel's solution of Einstein field equations \cite{Godel49}. However, such a solution is unstable to small perturbations and of course they are undesirable. Therefore it is reasonable to impose constraints that ban them from physical consideration such as NCL. (Thanks to Prof. Hall for pointing out this.)
\end{thebibliography}
\end{document}